\documentclass[11pt]{article}

\usepackage[utf8]{inputenc} 
\usepackage{hyperref}       
\usepackage{url}            
\usepackage{booktabs}       
\usepackage{amsfonts}       
\usepackage{nicefrac}       
\usepackage{microtype}      
\usepackage[pdftex]{graphicx}
\newcommand{\dist}{\mathsf{dist}}

\newcommand{\E}{I\!\!E}
\newsavebox{\fmbox}
\newcommand{\eps}{\varepsilon}
\newtheorem{lemma}{Lemma}
\newtheorem{theorem}{Theorem}

\newenvironment{proof}{\noindent{\bf Proof : }\small }{\normalsize}

\begin{document}
\title{The content correlation of multiple  streaming edges}

%

\author{Michel de Rougemont  $^1$ \and 
Guillaume Vimont$^1$}

\date{%
    $^1$ University Paris II,  CNRS-IRIF\\%
    \today
}

\maketitle

\begin{abstract}
  We study how to detect clusters in a graph defined by a stream of edges, without storing the entire graph. We  extend the approach to dynamic graphs defined by the most recent edges of the stream and  to several streams.  The {\em content correlation }of two streams  $\rho(t)$ is the Jaccard similarity of  their  clusters  in the windows before time $t$. We propose a simple and efficient method to approximate this correlation online and  show that for dynamic random graphs  which follow a power law degree distribution, we can guarantee a good approximation.  As an application, we follow Twitter streams  and compute their content correlations online. We then propose a {\em search by correlation} where answers to sets of keywords  are entirely based on the small correlations of the streams. Answers are ordered by the correlations, and explanations can be traced with the  stored clusters.\\\\
\end{abstract}


\section{Introduction}

Consider a stream of edges of a graph which follows a power law degree distribution.  Sliding windows define dynamic graphs $G_t$ and we select each edge with a uniform probability in each window. There are several techniques which generalize the original Reservoir sampling  \cite{V85} for a fixed window, to  dynamic windows.
At any given time $t$, we have a Reservoir which keeps $k$ edges among the possible $m$ edges  which have been read ($m>>k$), and each edge has the same probability $k/m$ to be chosen in the Reservoir. If the cluster $S$ and the Reservoir size $k$ are  large enough,  a large connected component  will appear in the Reservoir. The random edges of $S$ are taken with the same probability $p=k/m$, i.e. follow the Erd\"{o}s-Renyi model $G(n,p)$ and the giant component in $S$ occurs if $p=k/m>1/n=1/|S|$, the classical phase transition probability.
In the dynamic case, we will observe changes in the communities as new communities appear and old communities disappear. At discrete times $t_i$ we store only the large connected components of the Reservoirs. 

We study the correlation between multiple streams of graph edges and define the {\em content correlation} $\rho$ of two streams of edges, based on the Jaccard similarity of their clusters and extend it to $\rho(t)$ on the dynamic graphs $G_t$.  We provide an analysis of dynamic random graphs which follow a power law degree distribution, based on the Configuration Model \cite{N10}.  We give sufficient conditions to detect clusters depending on their size, the Reservoir size $k$  and the length of the observation. We then approximate the content correlation with an online algorithm, i.e. estimate the correlation $\rho(t_i)$ at discrete times $t_i$.

Streams of graph edges are ubiquitous, in particular in social networks  where the graphs have a degree distribution close to a power law, a small diameter and clusters (communities) which evolve in time.  Twitter streams are defined by some tags and generate  a stream of edges of large dynamic  graphs. We estimate the correlations between multiple streams and obtain a correlation matrix $A$, from which we build a phylogeny tree.  We then introduce a   {\em Search by correlation}: given some tags, we find the  most correlated tags which depend  on the history of the clusters and on the phylogeny tree. Our main results are:
\begin{itemize}
\item  We propose an Algorithm which  detects  large static clusters in a graph which follows a power law degree distribution,  with high probability,  using only a few edge samples  (theorem \ref{t1}) and  its extension to dynamic graphs (theorem \ref{t2}), 

\item  We present an online algorithm for $\rho(t)$. For two dynamic graphs with clusters $S$ and $S'$ whose Jaccard similarity  is $\rho^*$, the correlation   $\rho(t)$ is close to $\rho^*$ (theorem \ref{t3}).
\end{itemize}

We estimated the content correlations of $4$ Twitter streams  (approximately $2.10^6$ edges) over $24$h online and built the closest phylogeny tree. The {\em Search by correlation}  illustrates the technique which uses the correlations, a phylogeny tree and the stored clusters.
In the second section, we define the framework of  dynamic graphs defined by a stream of edges and introduce the notion of content correlation of two streams. In the third section, we set a model of random dynamic graphs where we can guarantee the approximation of the correlation. In the fourth section, we present our analysis of multiple Twitter channels and their correlations. In the fifth section, we introduce the Search by correlation, using a phylogeny tree built from the correlation matrix of the streams.

\section{Dynamic graphs in a window of streaming edges}

Let  $e_1, e_2,....e_i...$ be a stream of edges where each $e_i=(u,v)$. It defines a graph $G=(V,E)$ where $V$ is the set of nodes and $E\subseteq V^2$ is the set of edges: we allow self-loops and multi edges and assume that the graph is symmetric. In the {\em window model } we isolate the most recent edges at some discrete $t_1, t_2,...$. 
There are two models of sliding windows: the most recent $\tau$ edges or the most recent edges in some fixed time interval $\tau$. If the rate (the number of edges per time unit) of the stream is fixed, both models coincide. It is not the case in practice, as the rate fluctuates within a factor $2$ at any given time.
 We take  the second model, i.e.  keep  the length of the window $\tau$, hence $t_1=\tau$ and  each $t_i=\tau +\lambda.(i-1)$ for $i>1 $ and $\lambda <\tau$ determines a window of length $\tau$ and a graph $G_i$ defined by the edges in the window or time interval $[t_i -\tau, t_i]$. 
 The number of edges in a window may  increase or decrease and reflects the increasing or decreasing rates of a stream. Consecutive windows overlap within a factor $\tau/\lambda$, about $50\%$ in the experiments. In practice, $\tau=60$ mins and $\lambda=30$ mins.
The graphs $G_{i+1}$ and
$G_{i}$ share many edges: old edges of $G_i$ are removed and new edges are added to  $G_{i+1}$. Social graphs have a specific structure, a specific degree distribution (power law), a small diameter and some dense clusters. The dynamic random graphs introduced in the next section satisfy these conditions.

There are several definitions of a cluster or community or dense subgraph. We consider a cluster of domain $S$ as a {\em maximal dense subgraph} which depends on a parameter $\gamma$. Let $\gamma \leq 1$ and let $E(S)$ be the multiset of internal edges i.e. edges $e=(u,v)$ where $u,v \in S$. A $\gamma$-cluster is maximal subset $S$ such that $E(S) \geq \gamma.|S|^2$.
There are several  other possible definitions of clusters which capture  the high internal global density, and there are many algorithms to detect such clusters in a static graph. We are mainly interested in the approximate detection of clusters in the dynamic case, without storing the whole graph.

\subsection{Detecting clusters in a stream of edges}

If the entire graph is known, there are several classical  techniques to approximate communities. For  $\gamma=1$, the problem is known as {\em Maxclique}, which  is NP-hard. In our framework, we do not store the entire graph but only a few edges, and will only approximate the communities. We take a uniform sampling of the edges\footnote{Notice that such a uniform sampling on the edges is equivalent to a sampling of the nodes proportionally to their degrees.} for each window of the stream $i$ which defines $G_i(t)$, and keep  $k$ samples in Reservoirs $R_i(t)$ for a  fixed size $k$. There are several  techniques to build such dynamic Reservoirs \cite{B02}.

The sampling method we propose is not new,  it is one of the sampling methods in  \cite{Ah13}. Its analysis for dynamic graphs which follow a power law degree distribution  is one of the central points of this paper. 
If a cluster is large, there will be a large connected component in the Reservoir as a witness. We ignore all the small connected components of the Reservoir and only store in a database the large ones.
In a typical experiment $k=400$ whereas we read $10^4$ edges in a window. We ignore all the components of size less than $10$, a threshold  value. 
For a graph $G_i$, there could be several large clusters $C_{i,j}$ or there could be none. For a stream $G_i(t)$, we write 
$C_{i,j} (t)$ for the $j$-th cluster of the stream $i$ at time $t$.

\subsection{Correlations}

The classical correlation, also called the {\em Pearson correlation} $p(X,Y)$ of  two random variables $X,Y$ of mean $\mu$ and standard deviation $\sigma$ is $\frac{\E[(X-\mu_X)(Y-\mu_Y)]}{\sigma_X.\sigma_Y}$. How could it be extended to graphs?  
One could choose some statistics on the graphs and take the correlations between the statistical parameters. Social graphs have however very similar statistics and yet  the core of the information seems to hide in the structure of their clusters. 
Given two graphs 
$G_1$ and 
$G_2$, a first approach to their content correlation would be to consider the Jaccard similarity\footnote{The Jaccard similarity or Index between two sets $A$ and $B$ is $J(A,B)=|A \cap B|/|A \cup B|$. The  Jaccard distance is $1-J(A,B)$.} $J(V_1,V_2)$  on the domains of the two graphs. It has several drawbacks: it is independent of the structures of the graphs, it is very sensitive to noise, nodes connected with one or few edges and it is not well adapted when the sizes are very different.  It also requires to store the entire graphs.

We propose instead the following approach: we first estimate the dense components (clusters or communities) using the uniform sampling on the edges in the sliding windows and  apply the Jaccard similarity only to these large dense components. 
It exploits the structure of the graphs, is insensitive to noise and adapts well when the graphs have different sizes.  Let $C_i=\bigcup_j C_{i,j}$ be the set of clusters of the graph 
$G_i$, for $i=1$ or $2$. The {\em correlation of two graphs}
$\rho =J(C_1,C_2)$.
This definition is $\rho(t_1)$ for the first window or for two static graphs. For two dynamic streams $G_1 (t)$ and 
$G_2(t) $which share a time scale, we generalize  $C_i$ to  $C_i(t)=\bigcup_{t'_\leq t}\bigcup_j C_{i,j}(t')$.  {\em The correlation of 
two graph streams $G_1 (t)$ and 
$G_2(t)$ is} $ \rho(t) =J(C_1(t),C_2(t)) $.  We can refine the correlation and define an {\em amortized correlation} $\rho_a(t)$, to give more importance  to the recent components.
In the next section, we  give an algorithmic solution for graphs presented as streams of edges which scales when we consider dynamic graphs. 

\subsection{What is stored over time}
At some discrete times $t_1, t_2,....$, we store the large connected components of the Reservoirs $R_t$. There could be none.  We 
use a NoSQL database,  with $4$  (Key, Values) tables where  
the key is always a tag (@x or \#y) and the Values store the clusters nodes. Notice that a stream is identified by a tag (or a set of tags) and  a cluster is also identified by a tag, its node of highest degree.
{\small 
\begin{itemize}
\item  {\em Stream(tag, list(cluster, timestamp))} is the table which provides the most recent clusters of a stream, 
\item {\em Cluster(tag, list(stream, timestamp, list(high-degree nodes), list(nodes,degree))))} is the table which provides the list of high-degree nodes and the list of nodes with their degree, in a given  cluster,  
\item  {\em Nodes(tag, list(stream, cluster, timestamp))} is  the table which provides for each node the list of streams, clusters and timestamps where the node appears, 
\item  {\em Correlation((tag1,tag2),  list(value,timestamp))} is  the table which provides for each pair of streams {\em (tag1,tag2)} the different correlation values $\rho(t)$.
\end{itemize}
}
\subsection{Other approaches}
There are many other approaches to detect clusters in streams of graphs edges. 
The {\em dynamic graphs algorithms}  community studies the compromise between update and query time in the worst case. The {\em graph streaming  approach}  \cite{M2014} emphasizes the space complexity  in the worst case and in particular for  the window model.  The {\em network sampling  approach} such as \cite{Ah13}  does consider the uniform sampling on the edges but  there is no analysis for the detection of clusters in dynamic graphs.  The detection of a {\em planted clique} is a classical problem  \cite{F2000}, hard when the clique size is for example $O(\sqrt(n)/2)$ in the worst case. The {\em graph mining community }  \cite{A10,A13,Z14}  studies the detection of clusters when the graphs are entirely known.

In our approach, we only consider classes of graphs which follow a  power law degree distribution, and study approximate algorithms for the detection of dynamic  $\gamma$-cliques in the window model using only small Reservoirs.  

\section{Deciding properties and correlations in dynamic random models}
We define a model of dynamic random graphs $G(t)$  which may or may not have clusters and follow a degree distribution using the Configuration Model.
Temporal Logic is a framework to decide temporal properties of  the dynamic graphs $G_t$.  Let $P$  be a graph property such as Connectivity, or the existence of a $\gamma$-cluster of size at least $10$. A typical temporal property is  $\Diamond ~P$ stating that there exists a $t$ such that $G_t \models P$ or  $\Box P$ stating that for all $t$,  $G_t \models P$. 
In this section, we show that the algorithmic approach guarantees a good approximation of the correlation $\rho(t)$ with high probability.
 
\subsection{Dynamic Random graphs}
The classical Erd\"{o}s-Renyi model $G(n,p)$ \cite{E60}, generates  random graphs with $n$ nodes and edges are taken independently with probability $p$ where $0<p<1$.  The degree distribution is close to a gaussian centered on $n.p$.  
 Most of the social graphs have a degree distribution ${\cal D}$ close to a power law, such as a Zipfian distribution distribution where $Prob[d=j]=c/j^2$, where $d$ is the degrre of a node. In this case, the maximum degree is  $d_{max}=O(\sqrt{n})$.  The Configuration Model  for ${\cal D}$ and  a graph with $n$ nodes enumerates each node $u$ with $d$ half-edges (stubs) and  takes a symmetric random matching $\pi$ between two stubs, for example with a uniform permutation such that 
$\pi(i) \neq i$. All the possible graphs are obtained with a distribution close to the uniform distribution.

A classical study is to find sufficient conditions so that a random graph has a {\em giant component}, i.e. of size $O(n)$ for a graph of size $n$. In  the Erd\"{o}s-Renyi model $G(n,p)$, it requires that $p>1/n$, and in the Configuration Model it requires that $\E[{\cal D}  ^2]-2\E[{\cal D}]>0$ as proved in \cite{MR98}, which is realized for the Zipfian distribution.
There is a phase transition for both models. 
There are several possible extensions to dynamic random graphs. In our model, the Dynamics is exogenous and at any time chooses between the Uniform and the Concentrated Dynamics.
\subsubsection{Uniform Dynamics}

we generalize the Configuration Model in a dynamic setting.
Remove $q \geq 2$ random edges, uniformly on the set of edges of $G$, freeing $2.q$  stubs. Generate a new uniform matching on these hubs to obtain $G'$. The distribution of random graphs stays uniform.

\subsubsection{Concentrated Dynamics}
a typical graph generated by the Uniform Dynamics is not likely to have a large cluster. 
The {\bf $S$-concentrated Dynamics} fixes some a subset $S$  among the nodes of high degree. Remove $q \geq 2$ edges, uniformly on the set of edges of $G$, freeing $2.q$  stubs, as before.  A stub is in $S$ if its origin or extremity is in $S$.  With probability $80\%$,  match the stubs in $S$ uniformly in $S$. With probability $20\%$,  match the stubs in $S$ uniformly in $V-S$. This dynamics  will concentrate edges in $S$ and will create a $\gamma$-cluster after a few iterations with high probability, assuming the degree distribution is a power law. The distribution of graphs with a $\gamma$-cluster stays also uniform.
\subsubsection{General Dynamics}

a general Dynamics is a function which chooses at any given time, one of the two strategies: either a Uniform Dynamics or some  $S$-concentrated Dynamics for a fixed $S$.
An example is the  {\bf Step Dynamics}: apply the Uniform Dynamics first, then switch to the $S$-dynamics for a time period $\Delta$, and switch back to the Uniform Dynamics. In our setting, the Dynamics depends on some external information, which we try to approximately recover. Notice that during the Uniform Dynamics phase, there are no large components and we store nothing. For the step phase, we store some components which will approximate $S$.
More complex strategies could involve several clusters $S_1$ and $S_2$ which  may or may not  intersect. 

\subsection{Deciding a static property: there is a large  $\gamma$-cluster}
Let $R$ be the Reservoir of size $k$  after we read $m$ edges $e_1, e_2,....e_m$.
In this simple case, we first fill the Reservoir with $e_1, e_2,....e_k$. For $i>k$, we decide to keep $e_i$ with probability $k/i$ and if we keep $e_i$, we remove one of the edges (with probability $1/k$) to make room for $e_i$. Each edge $e_i$ has then probability $k/m$ to be in the Reservoir, i.e. uniform.

The probabilistic space $\Omega$ is determined by the choices taken at every step by  the Reservoir sampling.  Consider a clique $S$ in the graph: its image in the Reservoir is the set $G_S$ of internal edges $e=(u,v)$ in the Reservoir, where $u,v \in S$. Each edge of the clique $S$ is selected with constant probability $k/m$, so we are in the case of  the Erd\"{o}s-Renyi model $G(n,p)$ where $n=|S|$ and $p=k/m$. We know that the phase transition occurs at $p=1/n$, i.e. there is a giant component if $p>1/n$ and the graph is connected if $p\geq \log n/n$.

In the case of a  $\gamma$-clusters $S$ associated with the $S$-concentrated Dynamics,  the phase transition occurs at $p=1/\gamma.n$.  Let $V_S$ be the set of nodes of the giant component $G_S$ whose nodes are in $S$.  As it is customary for approximate algorithms, we write $Prob_{\Omega}[ {\rm Condition} ] \geq 1-\delta$ to say that the Condition is true with high probability.

\begin{lemma}\label{bigc}
For $m$ large enough, there exists $\alpha=O(\log n)$ and $\delta$ such that if $|S| \geq m/\gamma.k$ in the concentrated Dynamics, then 
$Prob_{\Omega}[ |V_S| > \alpha ] \geq 1-\delta$.
\end{lemma}
\begin{proof}
If $S$ is almost a clique, i.e. a $\gamma$-cluster, then the phase transition occurs at 
$p=1/\gamma.|S|$. Hence if $p > 1/\gamma.|S|$, there is a giant component of size larger than a constant times $|S|$, say $|S|/2$ with high probability $1-\delta$. As the probability of the edges is $k/m$, it occurs if  $|S| \geq m/\gamma.k$.  Hence for $m$ large enough, there exists  $\alpha=O(\log n)$ such that  $Prob_{\Omega}[ |V_S| >  \alpha ] \geq 1-\delta$.
\end{proof}

In order to decide the graph property $P$: {\em there is a large $\gamma$-cluster}, consider this simple algorithm.

{\bf Static Cluster detection  Algorithm 1}: let $C$ be the  largest connected component of the Reservoir $R$. If $|C| \geq  \alpha$ then Accept, else Reject.

\begin{theorem}\label{t1}
If $|S| \geq m/\gamma.k$  for the concentrated Dynamics, then:
 
$Prob_{\Omega}[ {\rm Algorithm ~1 ~Accepts} ] \geq 1-\delta$

and for the uniform Dynamics:

$Prob_{\Omega}[ {\rm Algorithm ~1 ~Rejects} ] \geq 1-\delta$.
\end{theorem}
\begin{proof}
 If $|S| \geq m/\gamma.k$ for the concentrated Dynamics, Lemma \ref{bigc} states that
 $|V_S| > \alpha$ with high probability, hence as $V_S \subseteq C$, the condition $|C| \geq \alpha$ is true with high probability hence $Prob_{\Omega}[ {\rm Algorithm ~1 ~Accepts} ] \geq 1-\delta$. For the Uniform Dynamics ($|S|=0$),  \cite{MR98} shows that the largest connected component has size $O(\log n)$.   Hence $Prob_{\Omega}[ {\rm Algorithm ~1 ~Rejects} ] \geq 1-\delta$.
\end{proof}

Notice that $m=c_1.n.\log n$, as the average degree in a power law is $c_1.\log n$.  If $k= \sqrt{c_1.n}.\log n $ and $|S| \geq m/\gamma.k=\sqrt{c_1.n}/\gamma$, it satisfies the condition and it can be realized with the nodes of high degree.

\subsection{Deciding a dynamic property: $\Diamond ~P$}
Let $P$ be the previous property: is there a $\gamma$-cluster? How do we decide $\Diamond ~P$?
Consider the step strategy of length $\Delta>\tau$. When we switch strategy at time $t_1$ there is a delay until $S$ is a $\gamma$-cluster and symmetrically the same delay when we switch again at time $t_2 >t_1$. The probabilistic space $\Omega_t$ is now much larger.

{\bf Dynamic Cluster detection  Algorithm 2}: let $C_i$ be the  largest connected component of a dynamic  Reservoir $R_i$ at time $t_i$. If  there is an $i$ such that $|C_i| \geq \alpha$, then  Accept, else Reject.

We can still distinguish between the Uniform and the $S$-concentrated Dynamics, if $S$ is large enough. Let $m(t)$  be the number of edges in the window at time $t$. Let $G(t)$ be a graph defined by a stream  of $m(t)$ edges following a power law  ${\cal D}$.

\begin{theorem}\label{t2}
 For the step Dynamics of length $\Delta$ and $t>t_2$,    if $|S| \geq m(t)/\gamma.k$, then 
$Prob_{\Omega}[ A_2 {\rm ~Accepts} ] \geq 1-\delta^{\Delta /\tau}$
For the Uniform Dynamics
$Prob_{\Omega}[ A_2 {\rm  ~Rejects} ] \geq (1-\delta) ^{\Delta /\lambda}$.
\end{theorem}
\begin{proof}
For each window, we can apply theorem \ref{t1} and there are $\Delta/\tau$ independent windows.
 If $|S| \geq m/\gamma.k$ for the concentrated Dynamics, the error probability is smaller than the error made for 
 $\Delta/\tau$ independent windows, which is $\delta^{\Delta/\tau}$. Hence 
  $Prob_{\Omega}[ {\rm Algorithm ~2 ~Accepts} ] \geq 1-\delta^{\Delta/\lambda}$. For the Uniform Dynamics (equivalent to $|S|=0$),  the algorithm needs to be correct at each $\Delta/\lambda$ step. Hence
    $Prob_{\Omega}[ {\rm Algorithm ~2 ~Rejects} ] \geq (1-\delta) ^{\Delta/\lambda}$.
\end{proof}

The probability to accept for the $S$ concentrated Dynamics is amplified whereas the probability to reject for the Uniform Dynamics decreases. One single error  generates a global error.
Clearly, we could also estimate $\Delta$, for step strategies with similar techniques.

\subsection{Correlation between two streams}\label{correlation}

Suppose we have two streams $G_1 (t)$ and 
$G_2(t) $ which share the same clock. Suppose that  $G_1 (t)$ is a step strategy $\Delta_1$ on a cluster $S_1$  and  $G_2 (t)$ is a step strategy $\Delta_2$ on a cluster $S_2$. Let $\rho^*=J(S_1,S_2)$.   How good is the estimation of their correlation? Let  $C_i(t)=\bigcup_{t' \leq t}\bigcup_j C_{i,j}(t')$ be the set of large clusters $C_{i,j}(t')$ at time $t' \leq t$ of the graph 
$G_i$, for $i=1$ or $2$. Consider the following online algorithm to compute $\rho(t)$:

{\bf Online  Algorithm 3 for $\rho(t)$.} At time  $t+\lambda$, compute the increase $\delta_i$ in size of $C_i(t+\lambda)$ for $ i=1,2$ from $C_i(t)$,    and $\delta'$ the increase in size of $C_1(t+\lambda)\cap C_2(t+\lambda)$. Suppose $\rho(t)=I/U$ where $I=|C_1(t)\cap C_2(t)|$ and $U=|C_1(t)\cup C_2(t)|$. Then:
$\rho(t+\lambda)= \rho(t)+\frac{U. \delta' -I.(\delta_1+\delta_2)}{U.(U+\delta_1+\delta_2)}$.

 A simple computation shows that $\rho(t+\lambda)= \frac{I+\delta' }{U+\delta_1+\delta_2}$, i.e. the correct definition. The $\delta_i, \delta'$ are computed by standard operations on sets.
\begin{figure}[!t]
  \centering
   \includegraphics[width=0.8\linewidth]{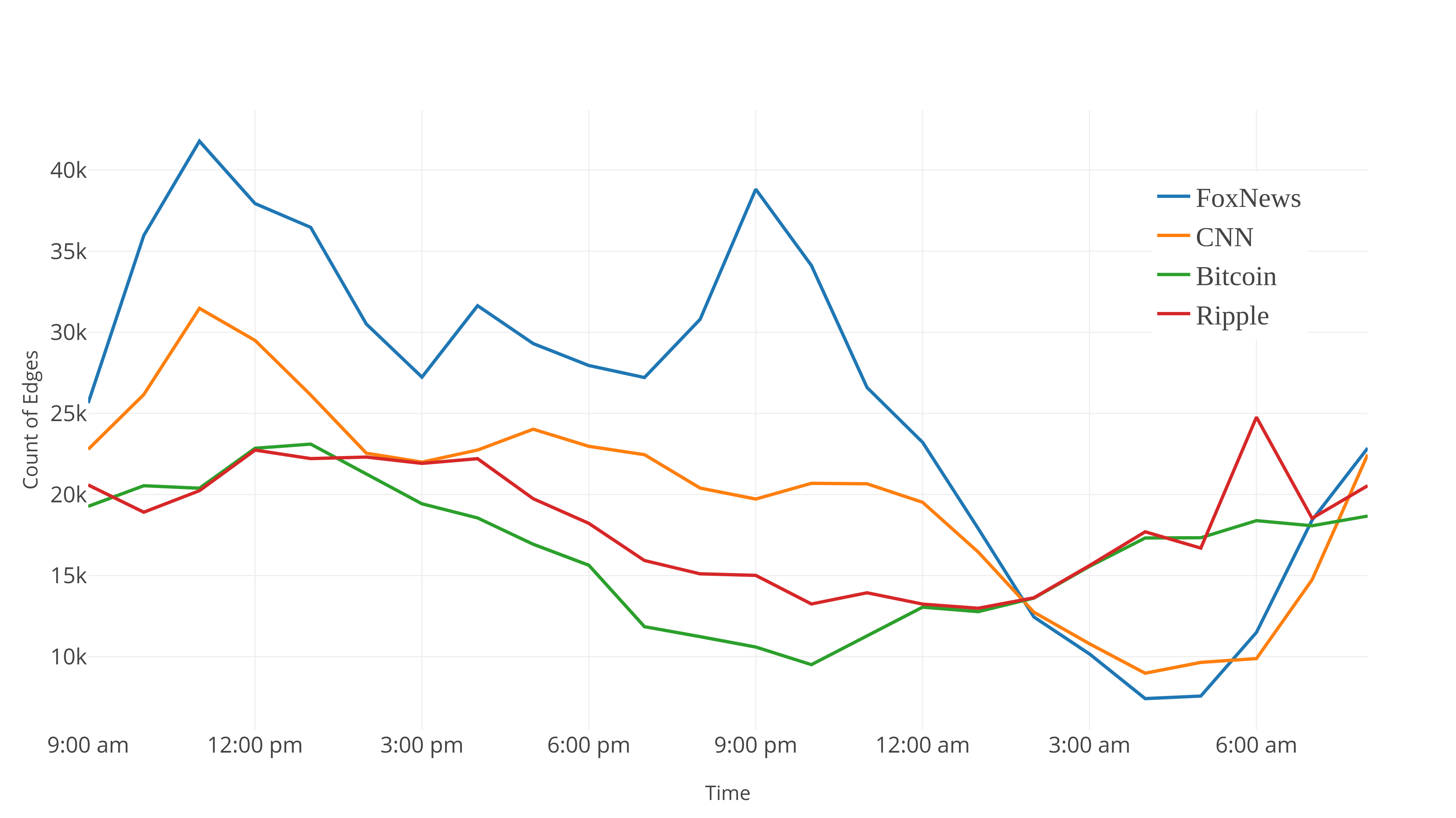}
  \caption{Number of edges  in $1h$ windows, for $4$ streams during $24h$\label{numb}}
\end{figure}
\begin{figure}[!t]
  \centering
   \includegraphics[width=0.8\linewidth]{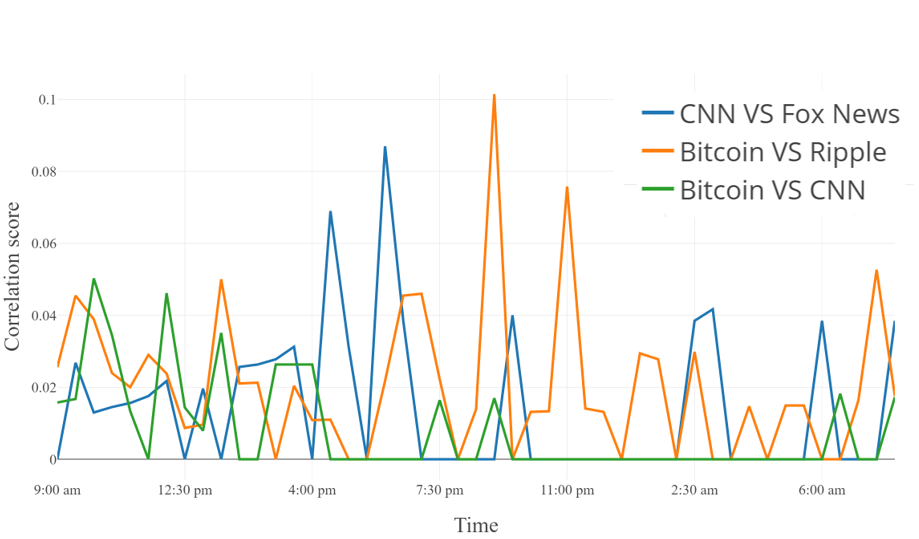}
  \caption{Online content  correlation for $24$h \label{cor1}}
\end{figure}
\begin{figure}[!t]
  \centering
   \includegraphics[width=0.8\linewidth]{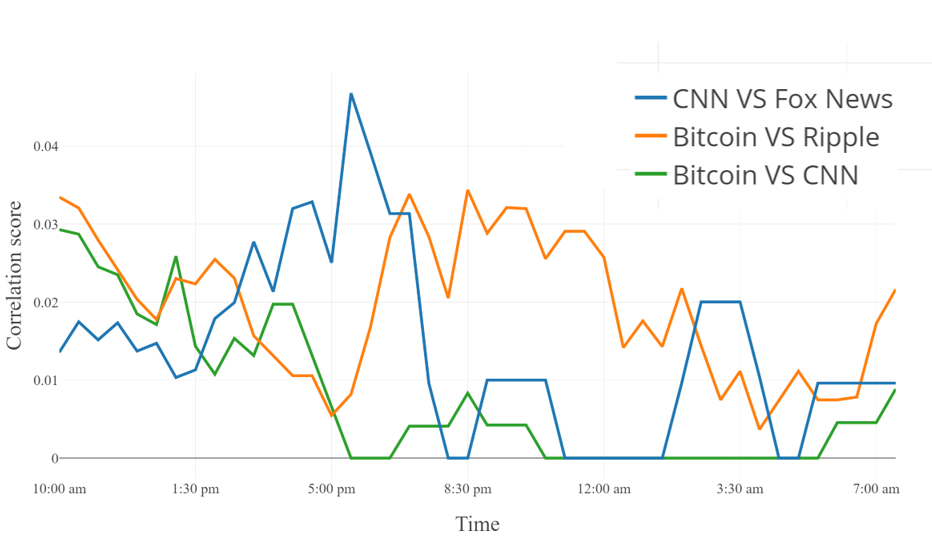}
  \caption{Online averaged content correlation for $24$h \label{cor2}}
\end{figure}
\begin{theorem}\label{t3}
Let $G_1 (t)$ and $G_2 (t)$ be two  step  strategies  before time $t$ on
two clusters such that $|S_i| \geq m/\gamma.k$ for $i=1,2$.
Then 
$Prob_{\Omega_t}[ |\rho(t)-\rho^*| \leq \eps ] \geq 1-\delta$.
\end{theorem}
\begin{proof}
After the first observed step, for example  on $S_1$, Lemma \ref{bigc} indicates that $V_{S_1}$ is already some approximation
of $S_1$. After $\Delta_1/\tau$ independent trials, $V_1=\bigcup_i V_{S_{1,i}}$ will be a good approximation of $S_1$. Similarly for $S_2$ and therefore $\rho(t)=J(V_1,V_2)$ will $(\eps,\delta)$ approximate $\rho^*$.
\end{proof}

\section{Twitter streams}
Given a set of tags such as \#CNN or \#Bitcoin, Twitter provides
 a stream of tweets represented as Json trees whose content  contains at least one of these tags. The {\em Twitter Graph} of the stream, is the graph $G=(V,E)$ with multiple edges $E$
where $V$ is the set of  tags $\#x$ or $@y$  seen and  for each tweet sent by  $@y$  which contains tags
$\#x$ ,$@z$ we construct the
edges  $(@y,\#x)$ and  $(@y,@z)$   in $E$.  The URL's which appear in the tweet can also be considered as nodes but we ignore them for simplicity. A stream of tweets is then transformed into a stream of edges $e_1,......e_m,....$.

We simultaneously captured $4$ twitter streams\footnote{Using a  program available on 
https://github.com/twitterUP2/stream  which takes some tags, a Reservoir size $k$, a window size $\tau$, a step  $\lambda$ and saves the large connected components of the Reservoirs $R_t$.}  on the tags 
\#CNN,  \#FoxNews, \#Bitcoin, and \#Xrp (Ripple) during 24 hours with a 
window size of $\tau=1h$ and a time interval $\lambda=30$mins, using a standard PC. Figure \ref{numb} indicates the number of edges seen in a window, approximately $m=20.10^3$ per stream, on $48$ points. For $24$ independent windows, we read 
approximately $48.10^4$ edges, and globally approximately $2.10^6$ edges. The Reservoirs size $k=400$ and on the average we save $100$ nodes and edges, i.e. $4.48.100\simeq 2.10^4$ edges, i.e. a compression of $100$. For $\gamma=0.8$, the minimum size of a cluster is  $m/\gamma.k \simeq 60$. Notice that $k$ is close to $\sqrt{m}$.

Figure \ref{cor1} gives the three mains correlations  $\rho(t)$ out of the possible $6$ and the averaged correlation: 
$$\rho'(t)=(\rho(t-1)+\rho(t)+\rho(t+1))/3$$

 The correlation is highly discontinuous, as it can be expected, but the averaged version is smooth. The maximum value is $1\%$ for the correlation and $0.5\%$ for the averaged version. It is always small as witnessed by the correlation matrix. We experienced very small changes in the correlations and $\rho'(t)$, also computed online, witnessed it.  The spectrum of the Reservoirs, i.e. the sizes of the large connected components is another interesting indicator. For the \#Bitcoin stream, there is a unique very large component.

\section{Search by correlation}
We stored the history of the large clusters for each stream, i.e. the set of nodes of the clusters. Given a search query defined by a set of tags,  the answer to the query is the set of the most correlated tags. We first need a definition of the correlation between a tag and a set of tags.   Given a stream, we need to find some other close streams and use the standard Phylogeny method.
\subsection{Phylogeny}
Given a correlation matrix 
$A_t$ between streams, a standard approach, such as the Neighbor Joining method 
\cite{NJ},   constructs a tree  $T$ with valued 
edges  
such that each stream appears as a leaf in $T$ and the distance $d(i,j)$ between two streams in the tree is approximately the distance defined by the correlation matrix, i.e.  $d(i,j) \simeq 1-A(i,j)$. This construction assumes an additive property of the distances, but there is always an approximate solution.

In a learning phase ended at time $t$, we construct the tree $T$. Later on,  we have a different matrix $A'$ and a different tree $T'$ as in the  Figure  \ref{phyl}.  The {\em Tree Edit distance with moves} is a standard distance between trees where the basic operations are: edition of a label, insertion/deletion of an edge, move of a subtree.
This distance is very easy to approximate
\cite{fmr2010}, although the exact distance is $NP$-hard to compute. We just have to compare the  $k$-grams of the subtrees at depth 
$k$.   For  the unordered trees of Figure  \ref{phyl}, $\dist(T,T')=2$. 

We can  hence easily detect small or large changes in the tree $T$. Given a stream, the neighbors of a leaf in $T$ are the closest streams in $T$, which we use in the Search by correlation.
\begin{figure}[ht]
  \centering
   \includegraphics[width=1.02\linewidth]{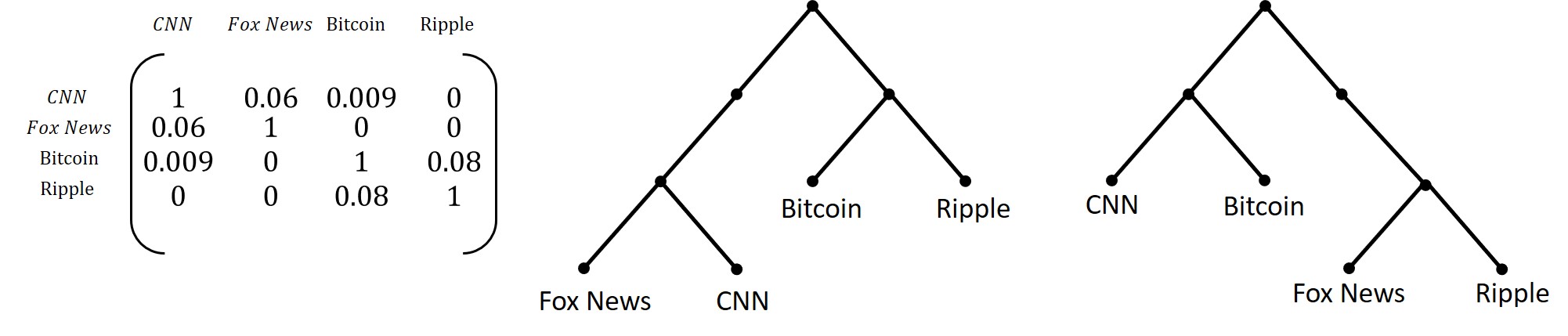}
  \caption{Phylogeny tree $T$ (center) from the correlation matrix, and another close tree $T'$ (right)\label{phyl}}
\end{figure}

We have a dynamic representation of the distances between streams, which we use in the next section.

\begin{center}
\begin{table*}[!t]
\renewcommand{\arraystretch}{1.3}
\centering
\resizebox{\textwidth}{!}{
\begin{tabular}{cccllcccllcc}
\cline{1-3} \cline{6-8} \cline{11-12}
\multicolumn{3}{|c|}{\textbf{Input: CNN}} &  & \multicolumn{1}{l|}{} & \multicolumn{3}{c|}{\textbf{Input: CNN POTUS}} &  & \multicolumn{1}{l|}{} & \multicolumn{2}{c|}{\textbf{Input: CNN POTUS NRA}} \\ \cline{1-3} \cline{6-8} \cline{11-12} 
\multicolumn{1}{|c|}{Ranking} & \multicolumn{1}{c|}{t =12} & \multicolumn{1}{c|}{t = 24} &  & \multicolumn{1}{l|}{} & \multicolumn{1}{l|}{Ranking} & \multicolumn{1}{c|}{t = 12} & \multicolumn{1}{c|}{t  = 24} &  & \multicolumn{1}{l|}{} & \multicolumn{1}{c|}{Ranking} & \multicolumn{1}{c|}{t = 12} \\ \cline{1-3} \cline{6-8} \cline{11-12} 
\multicolumn{1}{|c|}{1} & \multicolumn{1}{c|}{\#ODonnell} & \multicolumn{1}{c|}{\#POTUS} &  & \multicolumn{1}{l|}{} & \multicolumn{1}{c|}{1} & \multicolumn{1}{c|}{\#MondayMotivaton} & \multicolumn{1}{c|}{\#Tucker} &  & \multicolumn{1}{l|}{} & \multicolumn{1}{c|}{1} & \multicolumn{1}{c|}{\#Hannity} \\ \cline{1-3} \cline{6-8} \cline{11-12} 
\multicolumn{1}{|c|}{2} & \multicolumn{1}{c|}{\#NRA} & \multicolumn{1}{c|}{\#NRA} &  & \multicolumn{1}{l|}{} & \multicolumn{1}{c|}{2} & \multicolumn{1}{c|}{\#Hannity} & \multicolumn{1}{c|}{\#2A} &  & \multicolumn{1}{l|}{} & \multicolumn{1}{c|}{2} & \multicolumn{1}{c|}{\#2A} \\ \cline{1-3} \cline{6-8} \cline{11-12} 
\multicolumn{1}{|c|}{3} & \multicolumn{1}{c|}{\#POTUS} & \multicolumn{1}{c|}{\#ODonnell} &  & \multicolumn{1}{l|}{} & \multicolumn{1}{c|}{3} & \multicolumn{1}{c|}{\#NRA} & \multicolumn{1}{c|}{\#NRA} &  & \multicolumn{1}{l|}{} & \multicolumn{1}{c|}{3} & \multicolumn{1}{c|}{\#1A} \\ \cline{1-3} \cline{6-8} \cline{11-12} 
\multicolumn{1}{|c|}{4} & \multicolumn{1}{c|}{\#Obamacare} & \multicolumn{1}{c|}{\#Tucker} &  & \multicolumn{1}{l|}{} & \multicolumn{1}{c|}{4} & \multicolumn{1}{c|}{\#1A} & \multicolumn{1}{c|}{\#1A} &  & \multicolumn{1}{l|}{} & \multicolumn{1}{c|}{4} & \multicolumn{1}{c|}{\#Clinton} \\ \cline{1-3} \cline{6-8} \cline{11-12} 
\multicolumn{1}{|c|}{5} & \multicolumn{1}{c|}{\#MondayMotivaton} & \multicolumn{1}{c|}{\#2A} &  & \multicolumn{1}{l|}{} & \multicolumn{1}{c|}{5} & \multicolumn{1}{c|}{\#Tucker} & \multicolumn{1}{c|}{\#Hannity} &  & \multicolumn{1}{l|}{} & \multicolumn{1}{c|}{5} & \multicolumn{1}{c|}{\#ClintonFoundation} \\ \cline{1-3} \cline{6-8} \cline{11-12} 
\multicolumn{1}{l}{} & \multicolumn{1}{l}{} & \multicolumn{1}{l}{} &  &  & \multicolumn{1}{l}{} & \multicolumn{1}{l}{} & \multicolumn{1}{l}{} &  &  & \multicolumn{1}{l}{} & \multicolumn{1}{l}{}
\end{tabular}
}
\caption{Search results on inputs: $\sigma$=CNN,  $\sigma$=CNN, POTUS and  $\sigma$=CNN, POTUS, NRA }
\label{my-label}
\end{table*}
\end{center}

\begin{figure*}[!t]
  \centering
   \includegraphics[width=0.6\linewidth]{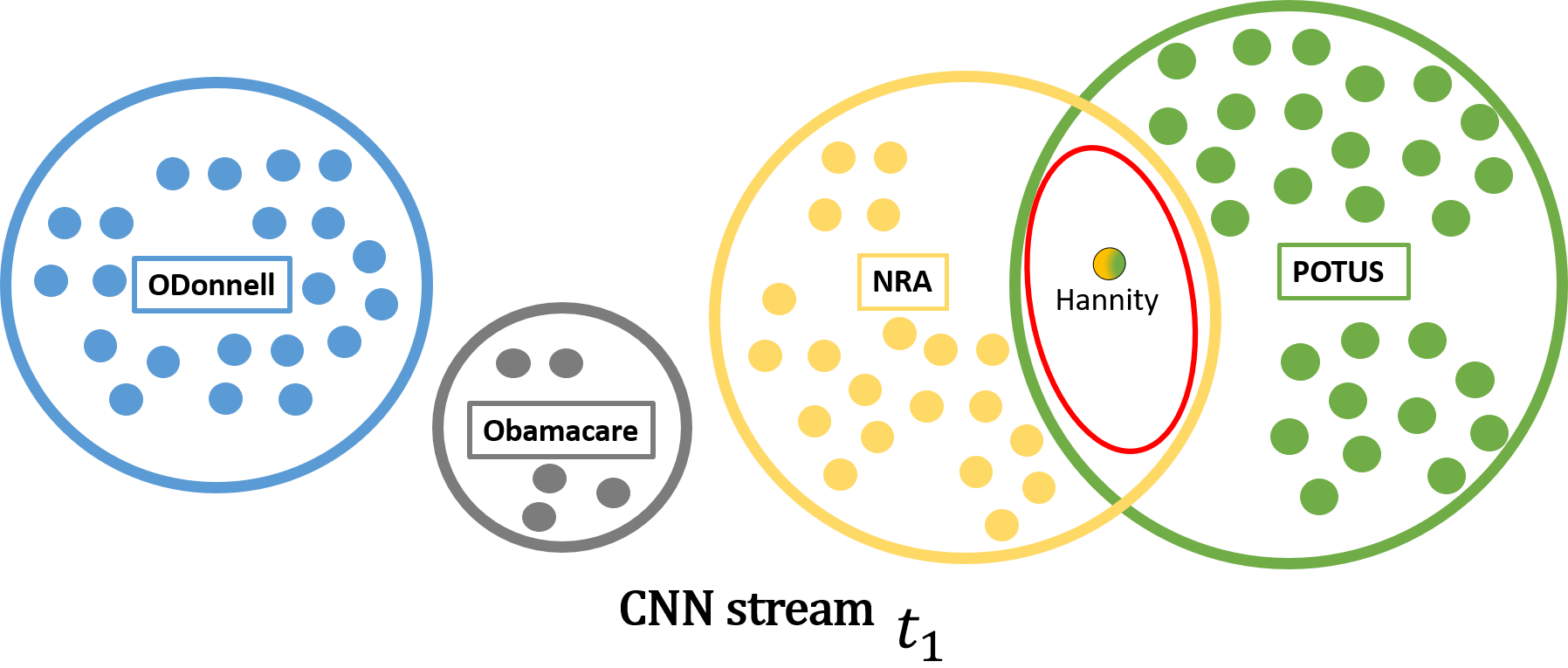}
  \caption{Nodes in an  intersection of clusters for the input CNN, POTUS, NRA \label{search_comp}}
\end{figure*}
\subsection{Correlation between tags at time $t$}
We extend the amortized correlation of section 2 to tags, i.e. node labels. Given a tag $\sigma$, we first check if it is a stream, the name of a cluster, or a simple node, using the tables {\em Stream, Clusters, Nodes}.  In each case we recover the most recent clusters  $C_{t_i}$  of a stream $\sigma'$, or of the nodes.  By construction, all the tags have at least one component they belong to. 

Given two tags $\sigma_1$ and $\sigma_2$, we retrieve their most recent components, $C_{1, t_i}$  and $C_{2, t_j}$ of the streams $\sigma_1'$ and $\sigma_2'$ and suppose $t_i>t_j$ and $\Delta(u)=t_{i}-t_{j}$. Assume $\dist$ is the distance between the streams $\sigma_1'$ and $\sigma_2'$  in the phylogeny tree. If  one tag belongs to  another component, we output the nodes of high-degree of the components. If 
these components intersect, then the tags of the intersection have a correlation coefficient   of
$(1- \frac{\Delta}{t}).(\frac{t_{i}}{t}) (1-\dist)$. If a tag is in several intersections, we add the correlations.

If the components  do not intersect, we look for another component $C'$ of another stream close  to $\sigma_1'$ and $\sigma_2'$ (using the phylogeny tree) which intersects  $C_{1, t_i}$  and $C_{2, t_j}$, or we look for older components $C_{1, t'_i}$  and $C_{2, t'_j}$ where $t'_i<t_i$ or $t'_j<t_j$ which do intersect. We generalize this definition for more than $2$ tags. We can synthetize the {\em Search Algorithm} as follows:

 {\bf Search by correlation Algorithm } $A_4(\sigma_1,....\sigma_l)$:
\begin{itemize}
\item  For each tag $\sigma_i$, find the most recent components $C_{i, t_i}$. Output the nodes with the highest correlation coefficient,
\item  If there  are no tags in the  intersections, look at close streams (using the phylogeny tree) and their components.
\end{itemize}

The nodes in some intersection of recent components will have the highest correlation and will be the in the top answers. The {\em explanation} of the search will be these clusters $C_{i, t_i}$, which are associated with the given tags. In the example for Figure \ref{search_comp}, the first answer \#Hannity belongs to 
to two CNN clusters, NRA and POTUS.

\subsection{Experimental results}
In the first example, we are given the tag $\sigma_1$=CNN. As the table \ref{my-label} indicates, the answers are the nodes of high-degree of the most recent component of the stream CNN. Notice that the anwers depend on $t$, for the two examples $t=12h$ and $t=24h$.  For the tags $\sigma_2$=CNN, POTUS,  we retrieve the most recent components and in this case, POTUS belongs a component of CNN: we output the nodes of high degree of that component.

 For the tag, $\sigma_3$=CNN, POTUS, NRA,  we retrieve the three most recent components and in this case, the component of POTUS intersects the component of NRA and are both components of CNN. We output the nodes of the intersection, ordered by degree.

\section{Conclusion}

We introduced {\em the content correlation $\rho$ between two graphs}  and its extension $\rho(t)$ between two {\em dynamical graphs} based on the Jaccard similarity of their clusters.  In the model of Uniform and Concentrated Dynamics for graphs with a power law degree distribution, we showed that the detection of a large connected component in a Reservoir built from  uniform edge samples is a good method to distinguish  a Uniform Dynamics from  a Concentrated Dynamics, when $S$ is large enough. This method generalizes to dynamic graphs and we can compute the content  correlation  of two streams with an online  algorithm (Algorithm 3).

As we read different streams of edges in the window model, we only store the large connected components  at some times $t_1, t_2,....$.  In our experiments, we followed $4$ Twitter streams for $24$h, reading $2.10^6 $ edges,  but kept approximately $2.10^4$ nodes, i.e. $1\%$ of the data.  From the correlation matrix, we obtained the closest phylogeny tree. We defined the {\em  Search by correlation} on some given tags, where the  answers are tags ordered by correlation. The witnessed  used as explanations of the correlations are the stored clusters.

\bibliographystyle{plain}
\bibliography{x1}

\end{document}